\documentclass[]{interact}

\usepackage{epstopdf}
\usepackage[caption=false]{subfig}
\usepackage{nccmath}
\usepackage{algpseudocode}
\usepackage{algorithm}
\usepackage{mathtools}
\usepackage{bbm}
\DeclarePairedDelimiter\ceil{\lceil}{\rceil}
\newcommand{\RN}[1]{\textup{\uppercase\expandafter{\romannumeral#1}}}

\usepackage[numbers,sort&compress]{natbib}
\bibpunct[, ]{[}{]}{,}{n}{,}{,}
\renewcommand\bibfont{\fontsize{10}{12}\selectfont}
\makeatletter
\def\NAT@def@citea{\def\@citea{\NAT@separator}}
\makeatother

\theoremstyle{plain}
\newtheorem{theorem}{Theorem}[section]
\newtheorem{lemma}[theorem]{Lemma}
\newtheorem{corollary}[theorem]{Corollary}

\newtheorem{note}{Note}
\theoremstyle{definition}
\newtheorem{definition}[theorem]{Definition}

\theoremstyle{remark}

\begin{document}


\title{Distributed Maximal Independent Set on Scale-Free Networks}

\author{
\name{Hasan Heydari\textsuperscript{a}\thanks{CONTACT Hasan Heydari Email: h\underline{ }heydari@ut.ac.ir}, S. Mahmoud Taheri\textsuperscript{a} and Kaveh Kavousi\textsuperscript{b}}
\affil{\textsuperscript{a}Faculty of Engineering Science, College of Engineering, University of Tehran, Tehran, Iran; \textsuperscript{b}Laboratory of Complex Biological Systems and Bioinformatics (CBB), Institute of Biochemistry and Biophysics (IBB), University of Tehran, Tehran, Iran}
}

\maketitle

\begin{abstract}
The problem of distributed maximal independent set (MIS) is investigated on inhomogeneous random graphs with power-law weights by which the scale-free networks can be produced. Such a particular problem has been solved on graphs with $n$ vertices by state-of-the-art algorithms with the time complexity of $O(\log{n})$. We prove that for a scale-free network with power-law exponent $\beta > 3$, the induced subgraph is constructed by vertices with degrees larger than $\log{n}\log^{*}{n}$ is a scale-free network with $\beta' = 2$, almost surely (a.s.). Then, we propose a new algorithm that computes an MIS on scale-free networks with the time complexity of $O(\frac{\log{n}}{\log{\log{n}}})$ a.s., which is better than $O(\log{n})$. Furthermore, we prove that on scale-free networks with $\beta \geq 3$, the arboricity and degeneracy are less than $2^{log^{1/3}n}$ with high probability (w.h.p.). Finally, we prove that the time complexity of finding an MIS on scale-free networks with $\beta\geq 3$ is⁡ $O(\log^{2/3}n)$ w.h.p.
\end{abstract}


\section{Introduction}\label{sec:introduction}
\subsection{Sale-Free networks}
Many real-world networks like the Internet, power grids, peer-to-peer networks, social and biological networks, the World Wide Web (WWW), research citation networks, etc. can not be described by classical Erd\H{o}s and R\'enyi (ER) \cite{Erdos1959} graphs \cite{Barabasi2016,Newman2003}. Indeed, they have a power-law degree distribution in the form of $P(k)\sim ck^{-\beta}$, where $P(k)$, distribution function of degree, denotes the fraction of vertices with degree $k$, $\beta$ is a constant exponent that describes a particular network and $c$ is a suitable constant \cite{Barabasi2002}. Such networks are called scale-free, and many models that explain their emergence have been proposed by some researchers. Two most important models in the class of scale-free networks are the Preferential Attachment \cite{Barabasi1999} and the Inhomogeneous Random Graphs \cite{Hofstad2016}. Inhomogeneous random graphs have multiple models, such as Chung and Lu \cite{Aiello2000,ChungLu2002}, Norros and Reittu \cite{Norros2006}, generalized random graphs \cite{Britton2006}, and Hofstad \cite{Hofstad2016}. It should be noted the Hofstad model is a generalization of the other models \cite{Friedrich2015}, so we use such model in this paper.

Important characteristics of the scale-free networks like the size of the giant component \cite{Bollob2007}, clustering coefficient \cite{Eggemann2011}, diameter \cite{Bollob2007,Reuven2003}, and the clique number \cite{Dole2017} have been considered. These networks have been examined in various contexts like finding parameterized cliques \cite{Friedrich2015}, PageRank \cite{Ningyuan2014}, information dissemination \cite{Fountoulakis2012}, and counting triangles, finding maximum cliques, transitive closure and finding perfect matching \cite{Brach2016}. In addition, scale-free networks have been applied to IoT networks \cite{Sohn2017} and WSNs \cite{Zheng2013,CYang2015,LLi2014} to enhance their synchronization, error tolerance, and robustness.

In other hand, Clauset et al. introduced a goodness-of-fit test based on Kolmogorov-Smirnov (KS) statistic to determine if the power-law distribution is a statistical plausible model for some continuous or discrete-valued data \cite{Clauset2009}. We will employ their method to evaluate if a set of numbers have the power-law distribution in this paper.
\subsection{Distributed MIS}
The problem of finding an MIS is one of the quite fundamental problems in the field of parallel and distributed computing because it solves the essential challenge of symmetry breaking, and furthermore, it is a building block for many distributed algorithms \cite{Barenboim2016,Censor2014,Liu2016,Yu2017}. More than 30 years ago, Alon et al. \cite{Alon1986} and Luby \cite{Luby1986} presented a simple randomized parallel algorithm to compute an MIS of a general graph. This algorithm computes an MIS for an $n$ node graph in $O(\log{n})$⁡ communication rounds with high probability (w.h.p.). Recently, Ghaffari \cite{Ghaffari2016} gave an MIS algorithm with the time complexity of $O(\log{\Delta} + 2^{O(\sqrt{\log{\log{n}}})})$ where $\Delta$ denotes the maximum degree of the graph. In addition, Barenboim and Elkin proposed an MIS algorithm running in $O(\log^{2}{\Delta} + 2^{O(\sqrt{\log{\log{n}}})}\big)$ on general graphs \cite{Barenboim2016}.

The MIS problem has been considered on some special family of graphs, as well. Working on random graph $G(n,p)$, for instance, shows that the MIS problem can be solved with the time complexity of $O\big(\ln{(np)} \ln{(\ln{p^{-1}})}\big)$ \cite{Krzywdzinski2015}. Studying the MIS problem leads to propose an algorithm on trees with the running time $O\big(\sqrt{\log{n} \log{\log{n}}}\big)$ \cite{Lenzen2011} and another algorithm on anonymous rings  with the time complexity of  $O(\sqrt{\log{n}})$ \cite{Fontaine2013}. Moreover, the MIS problem can be solved on bounded-independence \cite{Schneider2010} and constant degree \cite{Goldberg1988} graphs with the time complexity $O(\log^{*}{n})$. Furthermore, Panconesi and Rizzi proposed a deterministic algorithm to find an MIS with $O(\Delta^{2} + \log^{*}{n})$ rounds which is better than the Alon et al. and Luby's algorithms on sparse graphs \cite{Panconesi2001}.

The MIS problem is also solved on bounded arboricity graphs. For graphs with arboricity $a =\Omega (\sqrt{\log{n}})$, an MIS can be computed deterministically with the time complexity of $O\big(a\sqrt{\log{n}} + a\log{a}\big)$ \cite{Barenboim2010}. Recently, Barenboim and Elkin proposed an $O(\log^{2/3}{n})$-time MIS algorithm for graphs with arboricity up to $2^{\log^{1/3}{n}}$ \cite{Barenboim2016}.

In the MIS algorithms, in addition to the time complexity, the message complexity also have been investigated. M\'etivier et al. presented an MIS algorithm with the time complexity of $O(\log{n})$ in which the bit complexity per channel is $O(\log{n})$; in contrast, in the Alon et al. and Luby's algorithms, the bit complexity per channel are $(\log^{2}{n})$ \cite{Metivier2011}. Recently, an algorithm is presented by Jeavons et al. with the time complexity of $O(\log{n})$ rounds and the optimal expected message complexity of $O(1)$ single-bit messages that broadcast by each node \cite{Jeavons2016FeedbackFromNature}.

\subsection{Contributions}
In this paper, we study the problem of distributed maximal independent set on scale-free networks. Scale-free networks are modeled by inhomogeneous random graphs with power-law weights. In this regard, the main contributions of the paper are as follows:
\begin{itemize}
	\item
	We prove the induced subgraph which constructed by vertices with degrees larger than $\log{n}\log^{*}{n}$ from a scale-free network with power-law exponent $\beta > 3$ is a scale-free network with $\beta' = 2$, almost surely (a.s.).
	\item
	We propose a new algorithm which compute an MIS on scale-free networks in $O(\frac{\log{n}}{\log{\log{n}}})$-time a.s., so it is better than $O(\log{n})$-time state-of-the-art algorithms.
	\item
	We prove that on scale-free networks with $\beta \geq 3$, the arboricity and the degeneracy are less than $2^{\log^{1/3}{n}}$ w.h.p.
	\item
	We show the time complexity of finding an MIS for $\beta \geq 3$ is $⁡O(\log^{2/3}{n})$ w.h.p. and for arbitrary $\beta$ is $⁡O(\log^{2/3}{n})$ a.s.
\end{itemize}
\subsection{Structure of the paper}
In Section 2, we review some preliminaries and results. In Section 3, we prove that an MIS on scale-free networks can be computed for $\beta \geq 3$ with the time complexity of $⁡O(\log^{2/3}{n})$ w.h.p. and for arbitrary $\beta$ in $⁡O(\log^{2/3}{n})$ rounds a.s. In Section 4, we prove that an MIS on scale-free networks can be computed in $O\big(\frac{\log{n}}{\log{\log{n}}}\big)$ rounds a.s. A simulation experiment is provided in Section 5. Finally, conclusion remarks and a few words about the future works are provided in Section 6.

\section{Preliminaries}
In this paper, the synchronous message passing model of distributed computation is employed. The network is modeled by a simple undirected unweighted graph $G=(V,E)$, where nodes represent computational devices and edges represent bidirectional communication links. We assume each node knows the number of nodes, the maximum degree, and the power-law exponent of the scale-free network. In each synchronous round, nodes may perform arbitrary finite local computation and send (receive) a message with arbitrary length to (from) each neighbor. In other words, the message passing model is \textit{Local} model \cite{Peleg2000}. Throughout the paper, $n$ stands for the number of nodes in the network. The set of neighbors of a node $v\in V$ at distance one is denoted by $N(v)$, such that $N(v) = \{u| u\in V , (v,u)\in E\}$. The degree of $v \in V$ is denoted by $deg(v)=|N_{v}|$. We define $V(H)$, $E(H)$, and $deg_{H}(v)$ to be the set of vertices, the set of edges and degree of $v$, respectively, with respect to a graph $H$. Typically, $H$ is an induced subgraph of $G$. The induced subgraph $G_{[S]}$ is the graph whose vertex set is $S\subseteq V(G)$. 

There exist several measures of efficiency in distributed algorithms; here we focus on the running time, i.e. the number of rounds of distributed communication \cite{Sarma2015}. We call an event occurs a.s. and w.h.p. if its probability tends to $1$ as $n \rightarrow \infty$ and it has probability at least $1 - n^{-\Omega (1)}$, respectively. The function $\log^{*}()$ is defined recursively as: $\log^{*}{0} = \log^{*}{1} = \log^{*}{2} = 0$ and $\log^{*}{n} = 1 + \log^{*}\ceil{\log{n}}$  for $n > 2$ \cite{Schneider2010}. In the following, some definitions of graph theory will be reviewed. 
\begin{definition}
	(Maximal Independent Set) Given an undirected graph $G=(V,E)$, an independent set in $G$ is a subset of vertices $U \subseteq V$, such that no two vertices in $U$ are adjacent.  An independent set $U$ is called MIS if no further vertex can be added to $U$ without violating independence condition. 
\end{definition}
\begin{definition} 
	(Arboricity) Given an undirected graph $G=(V,E)$, the arboricity $a(G)$  is the smallest integer $k$ for which there exists forests $T_{1},T_{2},...,T_{k}$ which are subgraphs of $G$, such that their union is $G$. An equivalent definition formulated by Nash-Williams \cite{Nash1964} states that 
	\begin{ceqn}
		\begin{align*}
		a(G) = max \Big\{ \frac{|E(H)|}{|V(H)|-1}  \Big| H \subseteq G , |V(H)| \geq 2 \Big\}
		\end{align*}
	\end{ceqn}
\end{definition}
\begin{definition}
	(Degeneracy) Given an undirected graph $G=(V,E)$, the degeneracy $d(G)$ is the smallest integer $k$ such that every nonempty induced subgraph of $G$ contains a vertex with degree at most $k$. In other words
	\begin{align*}
	d(G) = max \big\{min\{deg_{H}(v) | v\in V(H)\} \big| H\subseteq G , H\neq \emptyset \big\}
	\end{align*}
\end{definition}
The degeneracy can be found by iteratively removing a vertex of minimum degree. Algorithm \ref{algDegeneracy} is a modified and simplified version of the algorithm that is proposed in \cite{Matula1983}, and will be needed in the next section. 
\begin{algorithm}
	\renewcommand{\baselinestretch}{1}
	\caption{Degeneracy Algorithm}\label{algDegeneracy}
	\begin{algorithmic}
		\State{\textbf{Input:} Graph $G=(V,E)$.}
		\State{\textbf{Output:} The output of this algorithm is $d(G)$.}
		\State{$d(G) = 0, n = |V|, H=G$}
		\For{$j =1 $ to $n$}
		\State{let $v$ be a vertex with minimum degree in $H$}
		\If{$d(G) < deg_{H}(v)$}
		\State{$d(G) = deg_{H}(v)$}
		\EndIf
		\State{$H = H_{[V(H) - \{v\}]}$} \Comment{Delete $v$ from $H$}
		\EndFor
	\end{algorithmic}
\end{algorithm}

Since the relationship between arboricity and degeneracy plays a vital role in our work, the following theorem is expressed. 
\begin{theorem}\label{arboricityDegeneracy}
	For any arbitrary graph $G$ with arboricity $a(G)$ and degeneracy $d(G)$ 
	\begin{ceqn}
		\begin{align*}
		a(G) \leq d(G)
		\end{align*}
	\end{ceqn}
\end{theorem}
\begin{proof}
	The proof is straightforward.
\end{proof}
In the following, we state some basic definitions and theorems about inhomogeneous random graphs, scale-free networks, and the MIS problem which are used in the subsequent sections from \cite{Barenboim2012,Barenboim2016,Friedrich2015,Hofstad2016,Krzywdzinski2015,Reuven2003}. 

Let us denote by $W=(W_{1},W_{2},...,W_{n})$ a random sample of size $n$ from a complementary cumulative distribution function (CCDF) $F(w)$ with observed values $w=(w_{1}^{o}, w_{2}^{o},...,w_{n}^{o})$  that has the empirical complementary cumulative distribution function (ECCDF) $F_{n}(w)=\Pr{[W\geq w]}$. We use the sequence $\sigma_{w} = w_{1},w_{2},...,w_{n}$ in which $w_{i} \leq w_{j}$ $(1\leq i < j\leq n)$ to denote the sorted order of $w$ \big($min(w) = w_{1}, max(w) = w_{n}$\big).
\begin{definition}\label{dfnPowerLawWeight}
	(Power-law weights) We say that $F_{n}(w)$ follows the power-law with exponent $\beta$, if there exist two constants $\alpha_{1}, \alpha_{2}$ such that
	\begin{ceqn}
		\begin{align*}
		\alpha_{1}w^{-\beta+1}\leq F_{n}(w) \leq \alpha_{2}w^{-\beta+1}
		\end{align*}
	\end{ceqn}
\end{definition}
\begin{definition}\label{dfnHIRG}
	(Hofstad inhomogeneous random graph) For $n \in N$ and an increasing sequence $\sigma_{w}$, the Hofstad inhomogeneous random graph $G(n,\sigma_{w})$ is a graph on vertex set $V$ $(|V|=n)$ such that each vertex $i$ has weight $w_{i}$, and the graph contains each edge $(i,j)$ with probability $p_{i,j}$ such that $p_{i,j}=\Omega (\frac{w_{i}w_{j}}{n})$ and $p_{i,j}=O(\frac{w_{i}w_{j}}{w_{i}w_{j}+n})$.
\end{definition}
\begin{note}
	If in Definition \ref{dfnHIRG}, the weights have a power-law distribution, then the degrees of vertices follows the power-law distribution, and we have a scale-free network.
\end{note}
Let us denote by $\textsl{g}_{sf}(\beta)$ the probability space of inhomogeneous random graphs with power-law weights in which the power-law exponent $\beta$ was created as described above. We will use $G_{sf}(\beta)$ to denote a graph drawn from $\textsl{g}_{sf}(\beta)$, and for $\beta \geq a$, we will use $G_{sf}(\beta \geq a)$. If $H$ is an induced subgraph of $G_{sf}(\beta)$ and a scale-free network, then we show its power-law exponent by $\beta_{H}$.
\begin{theorem}\label{thmWeightDegreeRela}
	Let $G=(V,E)$ be an inhomogeneous random graph. For $i \in V$, we have
	\begin{ceqn}
		\begin{equation*}
		E[deg(i)] = \Theta (w_{i})
		\end{equation*}
	\end{ceqn}
\end{theorem}
\begin{theorem}\label{lemEXi}
	Let $G_{sf}(\beta \geq 2)$ be an inhomogeneous random graph with power-law weights. The induced subgraph $H$ is constructed by vertices with weights larger than $w_{i}$ from $G_{sf}(\beta \geq 2)$. Then we have
	\begin{ceqn}
		\begin{align}
		& w_{i} = \Theta(1)(\frac{n}{n-i})^{1/(\beta -1)} \label{eq:1}\\
		& E[deg_{H}(i)] = O(w_{i}^{- \beta +3}) \label{eq:3}
		\end{align}
	\end{ceqn}
	According to the above relation and for $\beta \geq 3$, we have
	\begin{ceqn}
		\begin{align}\label{eq3LemEXi}
		E[deg_{H}(i)] = O(1)
		\end{align}
	\end{ceqn}
\end{theorem}
The relationships between graph $G$ and induced subgraph $H$ are shown in Table \ref{Tab:1}. Note that in graph $H$, for simplicity, the labels of vertices is not changed and is same as with the labels of $G$.
\begin{table*}[!htb]
	\centering
	\includegraphics[scale=0.42]{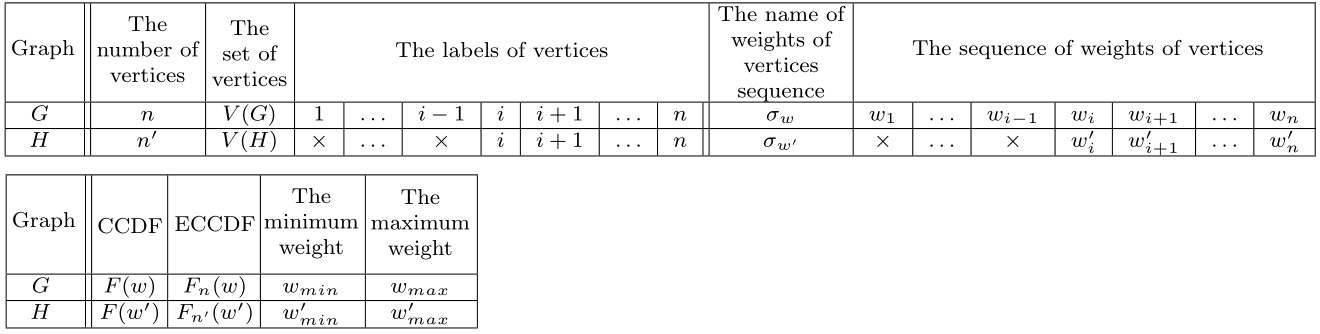}	
	\caption{The relationships between graph $G$ and induced subgraph $H$.}\label{Tab:1}
\end{table*}
\begin{theorem}\label{scalefreeDiam}
	The diameter of scale-free networks for $2< \beta <3$, $\beta = 3$, and $\beta >3$ is $O\big(\log{\log{n}}\big)$, $O(\frac{\log{n}}{\log{\log{n}}})$ and $O(\log{n})$ a.s., respectively.
\end{theorem}
\begin{theorem}\label{barenboimElkinArboriMIS}
	In any graph $G$ with bounded arboricity $a(G) \leq 2^{\log^{1/3}{n}}$, a maximal independent set can be computed in $O(\log^{2/3}{n})$ rounds w.h.p.
\end{theorem}
\begin{theorem}\label{barenElkinpolyLog}
	In any graph $G$ with $\Delta = poly(\log{n})$, the MIS problem can be computed in $O\Big(exp\big(\sqrt{\log{\log{n}}}\big)\Big)$ rounds w.h.p.
\end{theorem}
\begin{theorem}\label{localModelMIS}
	In any distributed network, the MIS problem can be computed in $O(diameter)$ rounds on Local model w.h.p.
\end{theorem}
\section{Computing an MIS in $O(\log^{2/3}{n})$ rounds}
In this section, first, we compute the degeneracy of graphs $G_{sf}(\beta\geq 3)$. Next, by using the result of computing the degeneracy, we compute the arboricity of graphs $G_{sf}(\beta\geq 3)$, as well. Then, we prove the MIS problem can be solved on $G_{sf}(\beta\geq 3)$ with the time complexity of $O(\log^{2/3}{n})$ rounds w.h.p. and on $G_{sf}(\beta)$ with the time complexity of $O(\log^{2/3}{n})$ rounds a.s.

For computing the degeneracy of $G_{sf}(\beta \geq 3)$, we propose Algorithm \ref{alg:MDegeneracy} which is the modified version of Algorithm \ref{algDegeneracy}. In each iteration of Algorithm \ref{algDegeneracy}, the vertex with minimum degree has been removed; in contrast, in Algorithm \ref{alg:MDegeneracy}, the vertex with minimum weight has been removed. We denote the output of Algorithm \ref{alg:MDegeneracy} with  modified degeneracy, $d_{w}(G)$. In what follows, first, we prove a lemma about modified degeneracy of $G_{sf}(\beta \geq 3)$. Then, we compute the degeneracy of $G_{sf}(\beta \geq 3)$. 
\begin{algorithm}
	\renewcommand{\baselinestretch}{1}
	\caption{Modified Degeneracy Algorithm Based on Weights of Vertices}\label{alg:MDegeneracy}
	\begin{algorithmic}
		\State{\textbf{Input:} Graph $G_{sf}(\beta)=(V,E)$ and $\sigma_{w}$}
		\State{\textbf{Output:} The output of this algorithm is $d_{w}(G)$.}
		\State{$d_{w}(G) = 0, n = |V|, H=G$}
		\For{$j=1$ to $n$}
		\State{let $v$ be the vertex with the weight $w_{j}$.}
		\If{$d_{w}(G) < deg_{H}(v)$}
		\State{$d_{w}(G) = deg_{H}(v)$}
		\EndIf
		\State{$H = H_{[V(H) - \{v\}]}$} \Comment{Delete $v$ from $H$}
		\EndFor
	\end{algorithmic}
\end{algorithm}
\begin{lemma}\label{lemDegeneracy}
	For each iteration of Algorithm \ref{alg:MDegeneracy}, the modified degeneracy, $d_{w}(G)$, for graphs $G_{sf}(\beta \geq 3)$ is less than $2^{\log^{1/3}{n}}$ w.h.p.
\end{lemma}
\begin{proof}
	It is sufficient to show that the degree of vertex $v$ for each iteration is less than $2^{\log^{1/3}{n}}$ w.h.p., i.e.
	\begin{ceqn}
		\begin{equation}\label{eq1ThmDegeneracy}
		P\{deg_{H}(v) \geq 2^{\log^{1/3}{n}}\}  \leq \frac{1}{n^{\omega (1)}}
		\end{equation}
	\end{ceqn}
	where $\omega (1)$ in Relation \eqref{eq1ThmDegeneracy} is strictly larger than 2.
	
	By Chernoff bound \cite{Dubhashi2009}, we have
	\begin{ceqn}
		\begin{align*}
		P\{X \geq a\}\leq \Big(\frac{exp(\frac{a}{E[X]}-1)}{(\frac{a}{E[X]})^{\frac{a}{E[X]}}}\Big)^{E[X]}
		\end{align*}
	\end{ceqn}
	Let us applying the above formula to Relation \eqref{eq1ThmDegeneracy}, we have
	\begin{ceqn}
		\begin{align*}
		P\{deg_{H}(v) \geq 2^{\log^{1/3}{n}}\} 
		& \leq \Big(\frac{e^{\frac{2^{\log^{1/3}{n}}}{E[deg_{H}(v)]} -1}}{(\frac{2^{\log^{1/3}{n}}}{E[deg_{H}(v)]})^{\frac{2^{\log^{1/3}{n}}}{E[deg_{H}(v)]}}} \Big)^{E[deg_{H}(v)]} 
		\\& \leq \Big(\frac{e^{\frac{2^{\log^{1/3}{n}}}{E[deg_{H}(v)]} }}{(\frac{2^{\log^{1/3}{n}}}{E[deg_{H}(v)]})^{\frac{2^{\log^{1/3}{n}}}{E[deg_{H}(v)]}}} \Big)^{E[deg_{H}(v)]} 
		\end{align*}
	\end{ceqn}
	and with simplifying the above relations, we obtain
	\begin{ceqn}
		\begin{align*}
		 P\{deg_{H}(v) \geq 2^{\log^{1/3}{n}}\} 
		 \leq \Big(\frac{e^{2^{\log^{1/3}{n}}}}{(\frac{2^{\log^{1/3}{n}}}{E[deg_{H}(v)]})^{2^{\log^{1/3}{n}}}} \Big) 
		& \leq \Big(\frac{\big(e\times E[deg_{H}(v)]\big)^{2^{\log^{1/3}{n}}}}{(2^{\log^{1/3}{n}})^{2^{\log^{1/3}{n}}}} \Big) \\
		&  \leq \Big(\frac{\big(O(1)\big)^{2^{\log^{1/3}{n}}}}{(2^{\log^{1/3}{n}})^{2^{\log^{1/3}{n}}}} \Big) \qquad \text{\big(By Rel. \ref{eq3LemEXi}\big)}
		\end{align*}
	\end{ceqn}
	Since $\Big(\frac{O(1)}{2^{\log^{1/3}{n}}} \Big)^{2^{\log^{1/3}{n}}} \leq \frac{1}{n^{\omega (1)}}$, then the proof is complete. 
\end{proof}
Let us explain the purpose of defining the modified degeneracy and proposing Lemma \ref{lemDegeneracy}. Consider the $j$th iteration of Algorithm \ref{algDegeneracy} and \ref{alg:MDegeneracy} in which $v$ is the vertex with minimum degree and weight, respectively. According to Theorem \ref{lemEXi}, Relation \ref{eq3LemEXi} can be applied to $E[deg_{H}(v)]$ in Algorithm \ref{alg:MDegeneracy}; on the other hand, this relation can not be applied to $E[deg_{H}(v)]$ in Algorithm \ref{algDegeneracy}. For this reason, we defined the modified degeneracy and proposed Lemma \ref{lemDegeneracy}. In the following theorem, by using Lemma \ref{lemDegeneracy}, we compute $d(G)$.
\begin{theorem}\label{thmDegeneracy}
	The degeneracy, $d(G)$, for graphs $G_{sf}(\beta \geq 3)$ is less than $2^{\log^{1/3}{n}}$ w.h.p.
\end{theorem}
\begin{proof}
	Let $\tilde{H}$ be an induced subgraph of $G_{sf}$ in the $j$th iteration of Algorithm \ref{algDegeneracy} in which $v$ be the vertex with minimum degree. Consider $u$ be the vertex with the minimum weight in $\tilde{H}$ and its weight in the sequence $\sigma_{w}$ (the sequence of weights in graph $G_{sf}$) be $w_{u}$. Let $H$ be another induced subgraph of $G$ that is constructed by vertices with weigths larger than $w_{u}$. By an argument presented in \cite{Friedrich2015}, we have
	\begin{ceqn}
		\begin{equation}\label{eq:thmDegeneracy1}
		P\{deg_{\tilde{H}}(v) \geq 2^{\log^{1/3}{n}}\}  \leq P\{deg_{H}(u) \geq 2^{\log^{1/3}{n}}\}
		\end{equation}
	\end{ceqn}
	And by applying Relation \eqref{eq1ThmDegeneracy} to Relation \eqref{eq:thmDegeneracy1}, for each iteration of Algorithm \ref{alg:MDegeneracy} we have
	\begin{ceqn}
		\begin{equation}\label{eq:thmDegeneracy2}
		P\{deg_{\tilde{H}}(v) \geq 2^{\log^{1/3}{n}}\}  \leq \frac{1}{n^{\omega (1)}}
		\end{equation}
	\end{ceqn} 
	Therefore, in each iteration of Algorithm \ref{algDegeneracy}, $deg_{\tilde{H}}(v)$ is less than $2^{\log^{1/3}{n}}$ w.h.p.
	
	What is left is to show that the degeneracy for these graphs is less than $2^{\log^{1/3}{n}}$ w.h.p., i.e.
	\begin{ceqn}
		\begin{align*}
		\Pr\{d(G) < 2^{\log^{1/3}{n}}\} \geq 1- n^{-\Omega (1)}
		\end{align*}
	\end{ceqn}
	We can now reformulate the above relation as follows
	\begin{ceqn}
		\begin{align*}
		 \Pr\{d(G) < 2^{\log^{1/3}{n}}\}
		&  = \Pr\{X_{1} < 2^{\log^{1/3}{n}} \cap ... ,\cap X_{n} < 2^{\log^{1/3}{n}}\} \\
		&  \geq 1- n^{-\Omega (1)}
		\end{align*}
	\end{ceqn}
	where $X_{j}$ $(1 \leq j \leq n)$ is a random variable refers to the degree of vertex $v$ in $j$th iteration of Algorithm \ref{algDegeneracy}. We have
	\begin{ceqn}
		\begin{align*}
		 \Pr\{X_{1} < 2^{\log^{1/3}{n}} \cap ... ,\cap X_{n} < 2^{\log^{1/3}{n}}\} 
		&    = 1- \Pr\{X_{1} \geq 2^{\log^{1/3}{n}} \cup ... ,\cup X_{n} \geq 2^{\log^{1/3}{n}}\} \\
		&    \geq 1 - \sum_{i=1}^{n}\Pr\{X_{i} \geq 2^{\log^{1/3}{n}}\} \\
		&    \geq 1- n\frac{1}{n^{\omega (1)}} 
		\end{align*}
	\end{ceqn}
	since $\omega (1)$ in Lemma \ref{lemDegeneracy} is strictly larger than 2, the proof is complete.
\end{proof}
\begin{corollary}\label{corArboricity1}
	With respect to Theorem \ref{arboricityDegeneracy} and \ref{thmDegeneracy}, the arboricity of graph $G_{sf}(\beta \geq 3)$ is less than $2^{\log^{1/3}{n}}$ w.h.p.
\end{corollary}
\begin{corollary}\label{corMISArboricity}
	According to Theorem \ref{barenboimElkinArboriMIS} and Corollary \ref{corArboricity1}, the time complexity of computing an MIS on graph $G_{sf}(\beta \geq 3)$ is $O(\log^{2/3}{n})$ w.h.p.
\end{corollary}
\begin{theorem}\label{thmLog23n}
	An MIS on scale-free networks with an arbitrary $\beta$ can be computed with the time complexity of $O(\log^{2/3}{n})$ a.s.
\end{theorem}
\begin{proof}
	We have divided the proof into two following cases:
	\begin{enumerate}
		\item
		$\beta < 3$) According to Theorem \ref{scalefreeDiam}, the diameter of graph $G_{sf}(\beta<3)$ is $O\big(\log{\log{n}}\big)$. Thus, with respect to Theorem \ref{localModelMIS}, the time complexity of computing a distributed MIS on these graphs is $O\big(\log{\log{n}}\big)$ a.s. Since $O\big(\log{\log{n}}\big) < O(\log^{2/3}{n})$, the proof for this case is complete.
		\item
		$\beta \geq 3$) With respect to Corollary \ref{corMISArboricity}, the proof for this case is straightforward as well.
	\end{enumerate}
\end{proof}
Now, according to Theorem \ref{thmLog23n}, we propose Algorithm \ref{algLog23n} to compute an MIS with the time complexity of $O(\log^{2/3}{n})$ rounds.
\begin{algorithm}
	\renewcommand{\baselinestretch}{1}
	\caption{The MIS Algorithm on $G_{sf}(\beta)$ in $O(\log^{2/3}{n})$ Rounds}\label{algLog23n}
	\begin{algorithmic}
		\State{\textbf{Input:} Graph $G_{sf}(\beta)=(V,E)$.}
		\State{\textbf{Output:} The output is an MIS.}
		\If{$\beta \leq 3$}
		\State{Run the trivial $O(diameter)$ MIS algorithm on $G_{sf}(\beta)$}
		\Else
		\State{Run Barenboim-Elkin MIS algorithm \cite{Barenboim2016} on $G_{sf}(\beta)$}
		\EndIf
	\end{algorithmic}
\end{algorithm}
\section{Computing an MIS in $O(\frac{\log{n}}{\log{\log{n}}})$ rounds} 
In this section, we describe an algorithm to compute an MIS with the time complexity of $O(\frac{\log{n}}{\log{\log{n}}})$ rounds on $G_{sf}(\beta)$ graphs. The main challenge in the problem of MIS on scale-free networks with the time complexity of $O(\frac{\log{n}}{\log{\log{n}}})$ is the case of $\beta >3$. The key idea of the algorithm for $\beta > 3$ is as follows. 

We have divided the algorithm into two separate phases. In the first phase, the induced subgraph $G_{\RN{1}}$ is constructed by vertices with degree larger than $\log{n}\log^{*}{n}$. We prove that the weights of vertices in $G_{\RN{1}}$ follows the power-law distribution a.s. and $G_{\RN{1}}$ is a scale-free network with power-law exponent $\beta' =2$ a.s. Thus, by using Theorem \ref{scalefreeDiam} and \ref{localModelMIS}, an MIS on $G_{\RN{1}}$ can be computed with the time complexity of $O(\log{\log{n}})$ rounds a.s. 

In the second phase, another induced subgraph $G_{\RN{2}}$ is constructed such that $V(G_{\RN{2}}) = \{v | v \in V-V(G_{\RN{1}}), v \notin N\big(MIS(G_{\RN{1}})\big) \}$, i.e. $G_{\RN{2}}$ consists of all vertices that are neither in $MIS$ nor in $N(MIS)$ after running the first phase. Since the maximum degree of $G_{\RN{2}}$ is $poly(\log{n})$, by using Theorem \ref{barenElkinpolyLog}, an MIS on $G_{\RN{2}}$ can be computed with the time complexity of $exp\big(\sqrt{\log{\log{n}}}\big)$ rounds w.h.p. 

In the following, in Theorems \ref{nDividlog}-\ref{thmIsScaleFree}, the weights of vertices is investigated, but in Theorem \ref{thmRelWeightDegLogLogStar} the degree of vertices is considered. Let us begin with the following theorem which describes how many of vertices in $G_{sf}(\beta)$ have weights greater than $\Theta (1) \log{n}$.
\begin{theorem}\label{nDividlog}
	In any graph $G_{sf}(\beta)$, the number of vertices with weight greater than $\Theta (1)\log{n}$ is $\Theta \big([\frac{n}{\log^{\beta -1}{n}}]\big)$.
\end{theorem}
\begin{proof}
	By Relation \eqref{eq:1}, we have
	\begin{align*}
	w_{\big([n-\frac{n}{\log^{\beta -1}{n}}]\big)} 
	 = \Theta (1) \Big(\frac{n}{n-(n-\frac{n}{\log^{\beta -1}{n}})} \Big)^{\frac{1}{\beta -1}} 
	= \Theta (1) \Big(\log^{\beta -1}{n} \Big)^{\frac{1}{\beta -1}} 
	 = \Theta(1) \log{n}
	\end{align*} 
	The weight of vertex $[n-\frac{n}{\log^{\beta -1}{n}}]$ (since $n-\frac{n}{\log^{\beta -1}{n}}$ is not necessarily a natural number, we use a brackets) in the sequence $\sigma_{w}$ is $\Theta(1) \log{n}$. Since $\sigma_{w}$ is an increasing sequence, there exist $[\frac{n}{\log^{\beta -1}{n}}]$ vertices with weight greater than $\Theta (1)\log{n}$ and the proof is complete.
\end{proof}
The following lemmas provide bounds for the expected degrees of vertices in graph $H$ which is an induced subgraph of $G_{sf}$. 
\begin{lemma}\label{lemZ2O}
	Let $G_{sf}(\beta \geq 3)=(V,E)$. The induced subgraph $H$ is constructed by vertices with weights larger than $w_{m}$ from $G_{sf}(\beta \geq 3)$.
	Then for degree of vertex $l \in V(H)$, we have 
	\begin{ceqn}
		\begin{align*}
		E[deg_{H}(l)]  =  O\Big( \frac{w_{l}}{w_{m}^{\beta - 2}} - \frac{w_{l}}{w_{m}^{\beta - 1}}\Big) 
		\end{align*}
	\end{ceqn}
\end{lemma}
\begin{proof}
	By inception of the idea given in \cite{Friedrich2015}, we have
\begin{align*}
	 E[deg_{H}(l)] 
	  = O\Big( \sum_{i=m}^{n}\frac{w_{i}w_{l}}{n} \Big) 
	&\quad  = O\Big( \frac{w_{l}}{n} \sum_{i=m}^{n}w_{i} \Big) \\
	&\quad  = O\Big( \frac{w_{l}}{n} \sum_{i=1}^{n} w_{i} \mathbbm{1}[w_{m} \leq w_{i}] \Big) \\
	&\quad  = O\Big( w_{l} E \big[W \mathbbm{1}[w_{m} \leq W]\big] \Big)\\
	&\quad  = O\Big( w_{l} E[W | w_{m} \leq W] \Pr[w_{m} \leq W] \Big)\\
	&\quad  = O\Big( w_{l} F_{n}(w_{m}) \int_{w_{1}}^{w_{n}} \Pr[w \leq W | w_{m} \leq W] dw \Big)\\
	&\quad  = O\bigg( w_{l} F_{n}(w_{m}) \Big(\int_{w_{1}}^{w_{m}} 1 dw 
	+ \int_{w_{m}}^{w_{n}} \frac{\Pr[w \leq  W]}{\Pr[w_{m} \leq W]} dw\Big) \bigg) \\
	&\quad  = O\bigg( w_{l} F_{n}(w_{m}) \Big((w_{m}-w_{1}) 
	+ \int_{w_{m}}^{w_{n}}\frac{F_{n}(w)}{F_{n}(w_{m})}dw \Big) \bigg)\\
	&\quad  = O\Big( w_{l} F_{n}(w_{m}) (w_{m} - w_{1}) 
	+ w_{l}F_{n}(w_{m})\int_{w_{m}}^{w_{n}}\frac{F_{n}(w)}{F_{n}(w_{m})}dw \Big) \\
	&\quad  = O\Big( w_{l} w_{m}^{-\beta +1} (w_{m} - w_{1}) + w_{l}\int_{w_{m}}^{w_{n}}F_{n}(w)dw \Big) \\
	&\quad  = O\Big( w_{l}w_{m}^{-\beta +1}(w_{m} - w_{1}) + w_{l}\big(\frac{w^{-\beta + 2}}{-\beta + 2}\big]_{w_{m}}^{w_{n}} \big) \Big) \\
	&\quad  = O\Big(w_{l}w_{m}^{-\beta +2} - w_{1}w_{l}w_{m}^{-\beta +1} + \frac{w_{l} w_{m}^{-\beta + 2} }{\beta -2} - \frac{w_{l}w_{n}^{-\beta +2}}{\beta -2} \Big) 
\end{align*}
	Since $w_{l} \leq w_{n}$, for $\beta \geq  3$, the value of the last term in the last relation is less than one. In addition, by Relation \eqref{eq:1}, we get $w_{1} = \Theta(1)$. From these, we concluded that
	\begin{ceqn}
		\begin{align*}
		E[deg_{H}(l)] = O\Big( \frac{w_{l}}{w_{m}^{\beta - 2}} - \frac{w_{l}}{w_{m}^{\beta - 1}}\Big) 
		\end{align*}
	\end{ceqn}
\end{proof}
\begin{lemma}\label{OmegaWei}
	Consider the assumptions in Lemma \ref{lemZ2O}. Then for degree of vertex $l \in V(H)$, we have 
	\begin{ceqn}
		\begin{align*}
		E[deg_{H}(l)]  =  \Omega \Big( \frac{w_{l}}{w_{m}^{\beta - 2}} - \frac{w_{l}}{w_{m}^{\beta - 1}}\Big) 
		\end{align*}
	\end{ceqn}
\end{lemma}
\begin{proof}
	\begin{align*}
	E[deg_{H}(l)] 
	& =  \Omega \Big( \sum_{i=m}^{n}\frac{w_{i}w_{l}}{w_{i}w_{l} + n} \Big) 
	 =  \Omega \Big( \sum_{i=m}^{n}\frac{w_{i}w_{l}}{n(1+\frac{w_{i}w_{l}}{n})} \Big) 
	 = \Omega \Big( \sum_{i=m}^{n}\frac{w_{i}w_{l}}{n\big(1+O(1)\big)} \Big) 
	\end{align*}
	From Relation \eqref{eq:1}, on $G_{sf}(\beta \geq 3)$ graphs, we obtain that $w_{n} \leq \Theta (\sqrt{n})$. Consequently, the last equality of the above relation holds. Therefore, we have
	\begin{align*}
	\qquad \qquad \qquad E[deg_{H}(l)]  =  \Omega  \Big( \sum_{i=m}^{n}\frac{w_{i}w_{l}}{n} \Big) 
	\end{align*}
	The remaining steps are exactly same as what we explained in the proof of Lemma \ref{lemZ2O}.
\end{proof}
By combining Lemma \ref{lemZ2O} and Lemma \ref{OmegaWei}, we can now state the following theorem about $E[deg_{H}(l)]$ which is a generalization of these lemmas.
\begin{theorem}\label{thm15}
	Consider the assumptions in Lemma \ref{lemZ2O}. Then for degree of vertex $l$ (that vertex with label $l$, weight $w_{l}$ in $G_{sf}$) in $H$, we have 
	\begin{ceqn}
		\begin{align*}
		\quad \qquad E[deg_{H}(l)]  =  \Theta \Big( \frac{w_{l}}{w_{m}^{\beta - 2}} - \frac{w_{l}}{w_{m}^{\beta - 1}}\Big) 
		\end{align*}
	\end{ceqn}
\end{theorem}
\begin{proof}
	According to Lemma \ref{lemZ2O} and Lemma \ref{OmegaWei}, the theorem holds.
\end{proof}
\begin{note}\label{note:Inhom}
	When the induced subgraph $H$ is constructed under the conditions in Lemma \ref{lemZ2O}, we can assume $H$ is still an inhomogeneous random graph. 
\end{note}
\begin{theorem}\label{thmWeightIsTheta1}
	Consider the assumptions in Lemma \ref{lemZ2O} and $w_{m} = \Theta (1) \log{n}$. Then for $\beta \geq 3$, $E[W'] = O(1)$ ($W'$ is a random variable that refers to weights of vertices in graph $H$.).
\end{theorem}
\begin{proof}
First of all, note that
\begin{align}\label{eq:7777}
		& \frac{1}{\sqrt{n-i}} \leq \frac{2}{\sqrt{n-i} + \sqrt{n-i-1}} \nonumber\\
		& \qquad \Rightarrow \frac{1}{\sqrt{n-i}} \leq 2(\sqrt{n-i} - \sqrt{n-i-1}) \nonumber\\
		& \qquad \Rightarrow \sum_{i=1}^{n-1} \frac{1}{\sqrt{n-i}} \leq 2 \sum_{i=1}^{n-1}(\sqrt{n-i} - \sqrt{n-i-1}) \nonumber\\
		& \qquad \Rightarrow \sum_{i=1}^{n-1} \frac{1}{\sqrt{n-i}} \leq 2 \sqrt{n-1} 
	\end{align}
	Next, let $w_{i}'$ be the weight of vertex $i$ in graph $H$. By Theorem \ref{nDividlog}, $\big|\{w_{m}',..., w_{n}'\}\big| = \Theta([\frac{n}{\log^{\beta -1}n}])$. We have
	\begin{align*}
	E[W']  
	& = \sum_{\mathclap{w_{i}'\in \{w_{m}',..., w_{n}'\}}} w_{i}'  \Pr(w_{i}' = W') \\
	&  = \sum_{\mathclap{w_{i}'\in \{w_{m}',..., w_{n}'\}}} w_{i}' \Theta\Big(\frac{1}{\frac{n}{\log^{\beta -1}n}}\Big) \qquad \text{(By Th. \ref{nDividlog})}
	\end{align*}
	\begin{align*}
	&  = \Theta\Big(\frac{\log^{\beta -1}n}{n}\Big) \sum_{\mathclap{w_{i}'\in \{w_{m}',..., w_{n}'\}}} w_{i}' \\
	&  = \Theta\Big(\frac{\log^{\beta -1}n}{n}\Big) \sum_{\mathclap{w_{i}'\in \{w_{m}',..., w_{n}'\}}} \Theta\big(E[deg_{H}(i)]\big) \quad \text{(By  Note \ref{note:Inhom} and Th. \ref{thmWeightDegreeRela})}\\
	&  = \Theta\Big(\frac{\log^{\beta -1}n}{n}\Big)  \sum_{\mathclap{i= n - [\frac{n}{\log^{\beta -1}n}] + 1}}^{n} \Theta \Big( \frac{w_{i}}{\log^{\beta - 2}n} 
	- \frac{w_{i}}{\log^{\beta - 1}n}\Big)  \qquad \text{(By Th. \ref{thm15})} \\
	& = \Theta\Big(\frac{\log{n} -1}{n}\Big) \sum_{\mathclap{i= n - [\frac{n}{\log^{\beta -1}n}] + 1}}^{n} \Theta (w_{i}) \\
	& = \Theta\Big(\frac{\log{n}}{n}\Big) \sum_{\mathclap{i= n - [\frac{n}{\log^{\beta -1}n}] + 1}}^{n-1} \Theta \big(\frac{n}{n-i} \big)^{\frac{1}{\beta -1}} \qquad \text{(By Rel. \ref{eq:1})}\\
	& \leq \Theta\Big(\frac{\log{n}}{n}\Big) \sum_{\mathclap{i= n - [\frac{n}{\log^{2}n}] + 1}}^{n-1} \Theta \big(\frac{n}{n-i} \big)^{\frac{1}{2}} \qquad (\beta \geq 3)\\
	& \leq \Theta\Big(\frac{\log{n}}{n}\Big) n^{\frac{1}{2}} \Bigg[ \sum_{\mathclap{i= 1}}^{n-1} \Theta \big(\frac{1}{n-i} \big)^{\frac{1}{2}} - \sum_{i=1}^{n - [\frac{n}{\log^{2}n}]} \Theta \big(\frac{1}{n-i} \big)^{\frac{1}{2}} \Bigg] \\
	& \leq \Theta\Big(\frac{\log{n}}{n}\Big) n^{\frac{1}{2}} \Bigg[(n-1)^{\frac{1}{2}} - \big(n - [\frac{n}{\log^{2}n}]\big)^{\frac{1}{2}} \Bigg] \qquad (\text{By Rel. \ref{eq:7777}})\\
	& = \Theta\Big(\frac{\log{n}}{n}\Big) n^{\frac{1}{2}} \Bigg[\frac{\frac{n}{\log^{2}n}-1}{(n-1)^{\frac{1}{2}} + \big(n - [\frac{n}{\log^{2}n}]\big)^{\frac{1}{2}}} \Bigg] \\
	& \leq \Theta\Big(\frac{\log{n}}{n}\Big) n^{\frac{1}{2}} \Bigg[\frac{\frac{n}{\log^{2}n}}{\big(n - [\frac{n}{\log^{2}n}]\big)^{\frac{1}{2}}} \Bigg] \\
	& = \Theta\Big(\frac{\log{n}}{n}\Big) n^{\frac{1}{2}} \Bigg[\frac{n^{\frac{1}{2}}}{\log^{2}n\big(1 - [\frac{1}{\log^{2}n}]\big)^{\frac{1}{2}}} \Bigg] \\
	& \leq \Theta\Big(\frac{\log{n}}{n}\Big) n^{\frac{1}{2}} \Bigg[\frac{n^{\frac{1}{2}}}{\log^{2}n(\frac{3}{4})^{\frac{1}{2}}} \Bigg] \\
	& = \Theta (\log^{-1}{n}) = O(1)
	\end{align*}
	$\Theta (\log^{-1}{n}) = O(1)$, so the proof is complete.
\end{proof}
\begin{theorem}\label{thmIsPowerLaw}
	Consider the assumptions in Lemma \ref{lemZ2O} and $\beta>3$. Then, the distribution of weights of vertices in graph $H$ follows the power-law distribution a.s.
\end{theorem}
\begin{proof}
	Let $W'$ be a random variable that refers to the weights in graph $H$. By Theorem \ref{thmWeightIsTheta1}, we have
	\begin{ceqn}
		\begin{align}\label{relEo}
		E[W'] = \Theta(\log^{3-\beta}{n} - \log^{2-\beta}{n} )
		\end{align}
	\end{ceqn}
	and by definition of expectation, we have
	\begin{ceqn}
		\begin{align}\label{eq:thm13.2}
		E[W'] = \int_{w_{min}'}^{w_{max}'}F_{n'}(w')dw'
		\end{align}
	\end{ceqn} 
	Combining Relation \eqref{relEo} with \eqref{eq:thm13.2} yields
	\begin{ceqn}
		\begin{align}\label{eq:thm13.3}
		\Theta(\log^{3-\beta}{n} - \log^{2-\beta}{n} ) = \int_{w_{min}'}^{w_{max}'}F_{n'}(w')dw'
		\end{align}
	\end{ceqn} 
	On the other hand, by Theorem \ref{thm15}
	\begin{ceqn}
		\begin{align}\label{relWmino}
		w_{min}'=\Theta(\log^{3-\beta}{n} - \log^{2-\beta}{n} )
		\end{align}
	\end{ceqn}
	Based on the above relations, if we define $F_{n'}(w')$ as follows, then Relation \eqref{eq:thm13.3} holds a.s.
	\begin{ceqn}
		\begin{align*}
		F_{n'}(w') = (w_{min}')^{\beta_{H} -1}(w')^{1-\beta_{H}}
		\end{align*}
	\end{ceqn}
	Next, the probability density function (PDF) of $F_{n'}(w')$ is obtained as follows
	\begin{ceqn}
		\begin{align}\label{relPdf}
		f(w') = c(w_{min}')^{\beta_{H} -1}(w')^{-\beta_{H}}
		\end{align}
	\end{ceqn}
	where $c$ is a constant. According to definition of continuous power-law distribution in \cite{Barabasi2016}, $f(w')$ and thus $F_{n'}(w')$ follow the power-law distribution. On the other hand, $(w_{min}')^{\beta_{H} -1} = O(1)$, hence by Definition \ref{dfnPowerLawWeight}, $F_{n'}(w')$ follows the power-law distribution. By Glivenko-Cantelli theorem \cite{DasGupta2008}, $F_{n'}(w')$ converge to $F(w')$ a.s. Therefore, $F(w')$ follows the power-law distribution, as well. Since the Regularity Conditions which is defined in \cite{Hofstad2016} are satisfied, the proof is complete.
\end{proof}
\begin{theorem}\label{thmIsScaleFree}
	Consider the assumptions in Lemma \ref{lemZ2O}. Let $w_{m} = \Theta (1) \log{n}$ and $\beta > 3$. Then, $H$ is a scale-free network with $\beta_{H} = 2$ a.s.
\end{theorem}
\begin{proof}
	By Definition \ref{dfnPowerLawWeight}, we have
	\begin{ceqn}
		\begin{align}\label{lbl4}
		\int_{w_{min}'}^{w_{max}'} F_{n'}(w') dw' = E[W']
		\end{align}
	\end{ceqn}
	\begin{ceqn}
		\begin{align}\label{lbl5}
		\int_{w_{min}'}^{w_{max}'} F_{n'}(w') dw' = \alpha w_{min}'^{- \beta_{H} +2}
		\end{align}
	\end{ceqn}
	Combining Relation \eqref{lbl4} with \eqref{lbl5}, we have
	\begin{align*}
	& \alpha w_{min}'^{- \beta_{H} +2} = E[W'] \\
	& \quad \Rightarrow w_{min}'^{- \beta_{H} +2} = O(1) \qquad \text{(By Th. \ref{thmWeightIsTheta1})}\\
	& \quad \Rightarrow \big(\frac{\Theta (1) \log{n}}{\log^{\beta -2}n} - \frac{\Theta (1) \log{n}}{\log^{\beta -1}n}\big)^{- \beta_{H} +2}  
	 = O(1)  \qquad \text{(By Note \ref{note:Inhom}, Th. \ref{thmWeightDegreeRela} and Th. \ref{thm15})}\\
	& \quad \Rightarrow (2-\beta_{H})\log{\big(\Theta (1) \log^{2 - \beta}{n} (\log{n} - 1)\big)} = O(1) \\
	& \quad \Rightarrow (2-\beta_{H})\log{\big(\Theta (1) \log^{2 - \beta}{n} \log{n}\big)} 
	 \leq  (2-\beta_{H})\log{\big(\Theta (1) \log^{2 - \beta}{n} (\log{n} - 1)\big)} 
	  = O(1) \\
	& \quad \Rightarrow (2-\beta_{H})\Big( \log{\big(\log^{3 - \beta}{n} \big)}  + \log{\Theta (1)}\Big) \leq O(1) \\
	& \quad \Rightarrow (\beta_{H} -2)(\beta -3)\log{\log{n}} \leq O(1) \\
	& \quad \Rightarrow \beta_{H}  \leq 2 \qquad \text{($n$ is sufficiently large)}
	\end{align*}
	By Theorem \ref{thmIsPowerLaw}, the distribution of weights of vertices in graph $H$ follows the power-law distribution a.s. Thus, $H$ is a scale-free network. In any scale-free network the power-law exponent is not less than 2, consequently $\beta_{H} = 2$ and the proof is complete.
\end{proof}
\begin{theorem}\label{thmRelWeightDegLogLogStar}
	Let $X$ be a random variable referring to the degree of vertex with weight $\Theta(1)\log{n}$ in a scale-free network with $\beta> 3$. Then, $X\leq \log{n}\log^{*}{n}$ w.h.p.
\end{theorem}
\begin{proof}
	By Chernoff bound \cite{upfal2005}, we have
	\begin{ceqn}
		\begin{align*}
		\Pr\big\{X\geq (1+\delta)E[X]\big\} \leq e^{-\frac{\delta E[X]}{3}}, \quad 1\leq \delta
		\end{align*}
	\end{ceqn}
	By setting $\delta = \frac{\log{n}\log^{*}{n}}{\log^{3-\beta }{n} - \log^{2-\beta }{n}} -1$, we get
	\begin{align*}
	& \Pr{\Big\{ X \geq 
		\big(1+ \frac{\log{n}\log^{*}{n}}{\log^{3-\beta }{n} - \log^{2-\beta }{n}} -1 \big)(\log^{3-\beta }{n} - \log^{2-\beta }{n}) \Big\}} \\
	& \qquad \leq exp \Big(-\frac{ \big(\frac{\log{n}\log^{*}{n}}{\log^{3-\beta }{n} - \log^{2-\beta }{n}} -1 \big)(\log^{3-\beta }{n} - \log^{2-\beta }{n})}{3} \Big) \\
	& \qquad \leq exp \Big(-\frac{ \log{n}\log^{*}{n} - \log^{3-\beta }{n} + \log^{2-\beta }{n} }{3} \Big)
	\end{align*}
	Since $\beta > 3$, $\log^{3-\beta }{n}$ and $\log^{2-\beta }{n}$ are less than 1. Thus, we have
	\begin{ceqn}
		\begin{align*}
		\Pr{\Big\{ X \geq \log{n}\log^{*}{n} \Big\}} \leq exp \Big(-\frac{ \log{n}\log^{*}{n} }{3} \Big) \leq \frac{1}{n^{\Omega (1)}}
		\end{align*}
	\end{ceqn}
	which completes the proof.
\end{proof}
Now, we state the main theorem of this section.
\begin{theorem}
	An MIS on scale-free networks can be computed in $O\big(\frac{\log{n}}{\log{\log{n}}}\big)$ rounds a.s.
\end{theorem}
\begin{proof}
	To provide the claim, we have divided the proof into three below cases:
	\begin{itemize}
		\item $\beta <3$) 
		By Theorem \ref{scalefreeDiam}, the diameter of scale-free networks with power-law exponent $\beta < 3$ is $O(\log{\log{n}})$, a.s. Thus, by Theorem \ref{localModelMIS}, the time complexity of computing an MIS on scale-free networks with power-law exponent $\beta < 3$ is $O(\log{\log{n}})$ a.s.
		\item $\beta =3$) 
		By Theorem \ref{scalefreeDiam}, the diameter of scale-free networks with power-law exponent $\beta = 3$ is $O\Big(\frac{\log{n}}{\log{\log{n}}}\Big)$ a.s. Thus, by Theorem \ref{localModelMIS}, the time complexity of computing an MIS on scale-free networks with power-law exponent $\beta = 3$ is $O\Big(\frac{\log{n}}{\log{\log{n}}}\Big)$ a.s.
		\item $\beta >3$) 
		By Theorem \ref{thmRelWeightDegLogLogStar}, the degree of the vertex with weight  $\Theta(1)\log{n}$ on $G_{sf}(\beta > 3)$ is at most $\log{n}\log^{*}{n}$ w.h.p.  We construct the induced subgraph $G_{\RN{1}}$ that is constructed by vertices with degrees larger than $\log{n}\log^{*}{n}$. By Theorem \ref{thmIsPowerLaw} the distribution of weights of vertices in $G_{\RN{1}}$ follows the power-law distribution. By Theorem \ref{thmIsScaleFree}, $G_{\RN{1}}$ is a scale-free network with power-law exponent $\beta' =2$. Thus, by case 1 of this theorem, an MIS on $G_{\RN{1}}$ can be computed in $O(\log{\log{n}})$ rounds a.s.
		
		In the next step, we construct another induced subgraph $G_{\RN{2}}$ such that $V(G_{\RN{2}}) = \{v | v \in V-V(G_{\RN{1}}), v \notin N\big(MIS(G_{\RN{1}})\big) \}$.
		It should be noted the maximum degree in $G_{\RN{2}}$ is $poly(\log{n})$. By Theorem \ref{barenElkinpolyLog}, an MIS can be computed in $exp\big(\sqrt{\log{\log{n}}}\big)$ w.h.p. on an arbitrary distributed network, when $\Delta = poly(\log{n})$. Thus, the running time for computing an MIS on graph $G_{\RN{2}}$ becomes $exp\big(\sqrt{\log{\log{n}}}\big)$ w.h.p.
	\end{itemize}
	According to the three studied cases, the complexity of computing an MIS on scale-free networks is 
	\begin{ceqn}
		\begin{align*}
		&max\Big\{O(\log{\log{n}}), O\Big(\frac{\log{n}}{\log{\log{n}}}\Big), 
		exp\big(\sqrt{\log{\log{n}}}\big)\Big\} 
		 = O\Big(\frac{\log{n}}{\log{\log{n}}}\Big)
		\end{align*}
	\end{ceqn}
\end{proof}

Eventually, we present our approach in pseudo-code in Algorithm \ref{algLogDLogLogn}.
\begin{algorithm}
	\renewcommand{\baselinestretch}{1}
	\caption{MIS Algorithm on $G_{sf}(\beta)$ in $O(\frac{\log{n}}{\log{\log{n}}})$ Rounds}\label{algLogDLogLogn}
	\begin{algorithmic}
		\State{\textbf{Input:} Graph $G_{sf}(\beta)=(V,E)$.}
		\State{\textbf{Output:} The output is an MIS.}
		\If{$\beta < 3$}
		\State{Run the trivial $O(diameter)$ MIS algorithm on $G_{sf}(\beta)$}
		\Else
		\Statex{Phase I}
		\State{$V_{\RN{1}} = \{v | v\in V , deg(v) \geq \log{n}\log^{*}{n} \}$}
		\State{$G_{\RN{1}} = G_{[V_{\RN{1}}]}$}
		\State{Run the trivial $O(diameter)$ MIS algorithm on $G_{\RN{1}}$}
		\Statex{Phase II}
		\State{$V_{\RN{2}} = \{v | v \in V - V_{\RN{1}} , v \notin N\big(MIS(G_{\RN{1}})\big) \}$}
		\State{$G_{\RN{2}} = G_{[V_{\RN{2}}]}$}
		\State{Run Barenboim-Elkin MIS algorithm \cite{Barenboim2012} on $G_{\RN{2}}$}
		\EndIf
	\end{algorithmic}
\end{algorithm}
\section{Experiments}
In this section, we evaluate two significant parts of our approach.
\begin{enumerate}
	\item
	Computing the degeneracy value for $G_{sf}(\beta \geq 3)$ graphs.
	\item
	Checking the induced subgraph constructed by vertices with degree larger than $\log{n}\log^{*}{n}$ from $G_{sf}(\beta \geq 3)$ is a scale-free network.
\end{enumerate}
For each $\beta \in \{3.0,3.5,4.0,4.5\}$, we have generated 10000 scale-free networks with $n_1=10000$ nodes, and 1000 scale-free networks with $n_2=100000$ nodes by using the proposed algorithm of Miller and Hagberg \cite{Miller2011}. Then, we have computed the degeneracy value of these networks. We show the results of this task in Fig. \ref{fig:subfigures1}. As shown in that Figure, it's clear that the degeneracy value of $G_{sf}(\beta \geq 3)$ graphs is less than $O(2^{log^{1/3} n})$ and the claim of Theorem \ref{thmDegeneracy} is confirmed. 
\begin{figure*}[!htb]
	\centering
	\subfloat[$\beta = 3.0$]{
		\label{fig:3.0}
		\includegraphics[scale=0.2]{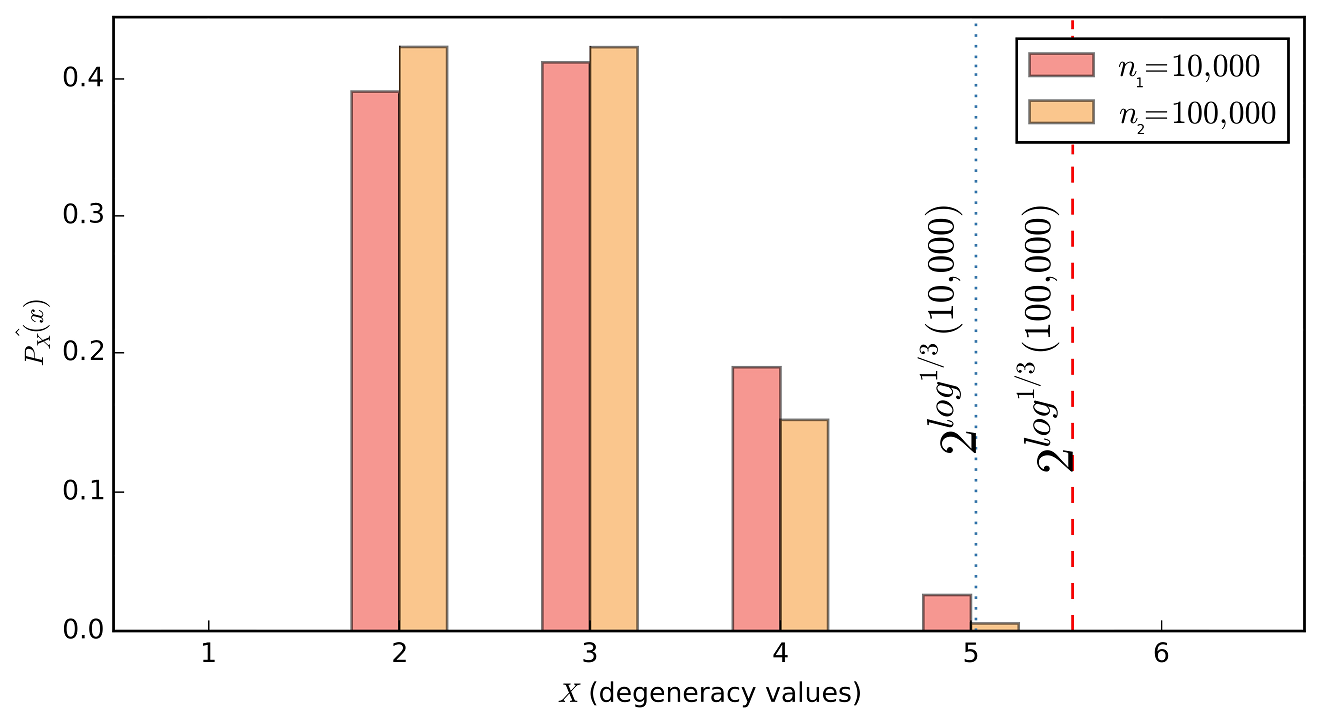}
	}
	\subfloat[$\beta = 3.5$]{
		\label{fig:3.2}
		\includegraphics[scale=0.2]{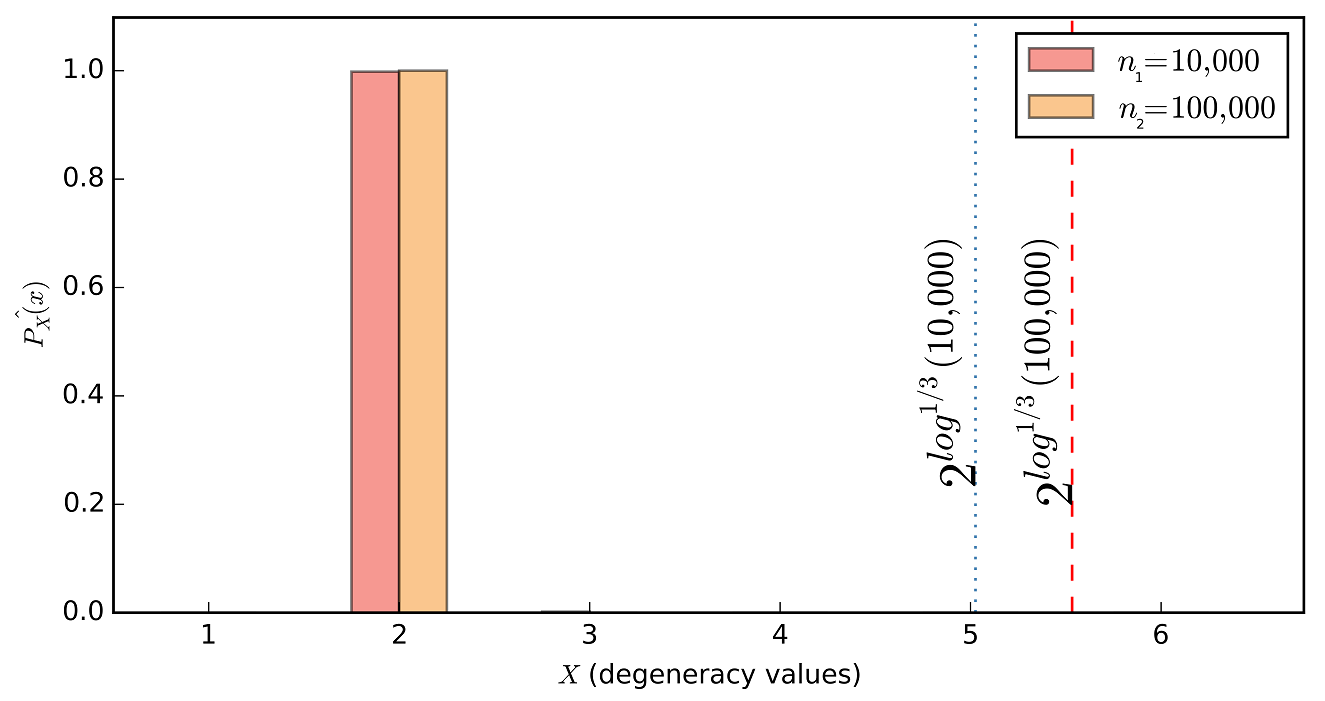}
	}\\
	\subfloat[$\beta = 4.0$]{
		\label{fig:3.5}
		\includegraphics[scale=0.2]{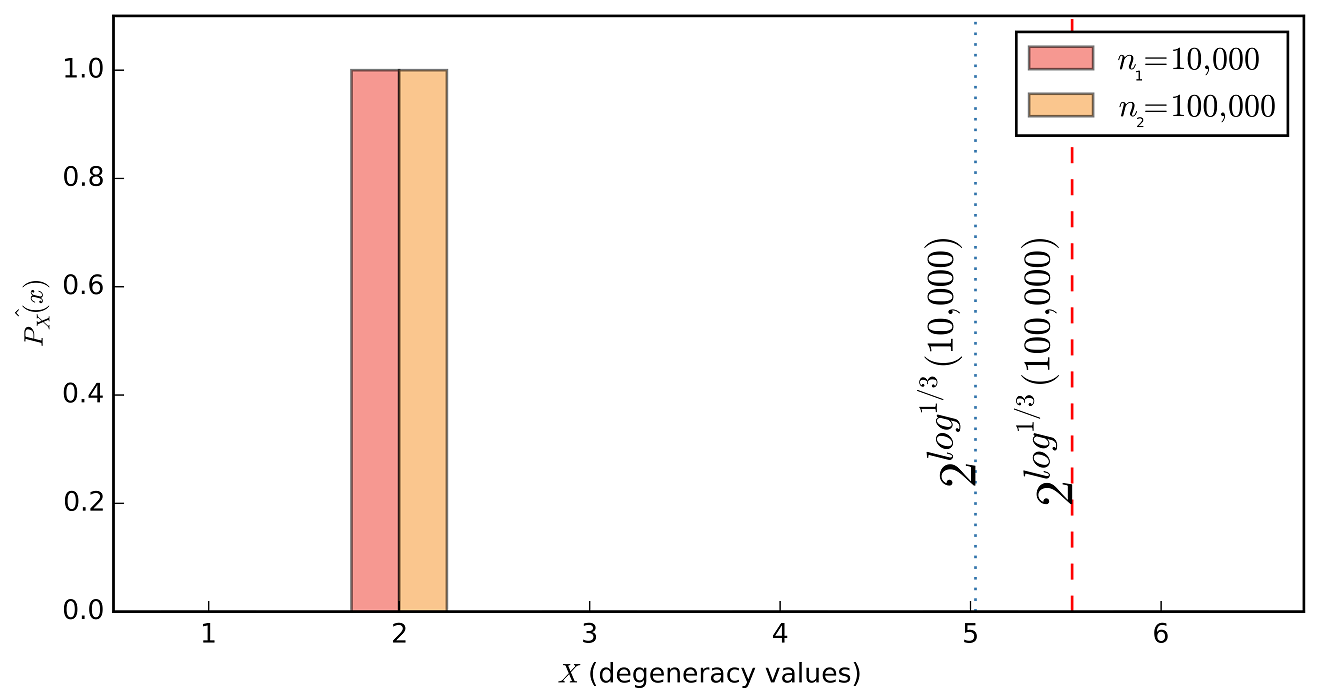}
	}
	\subfloat[$\beta = 4.5$]{
		\label{fig:4.0}
		\includegraphics[scale=0.2]{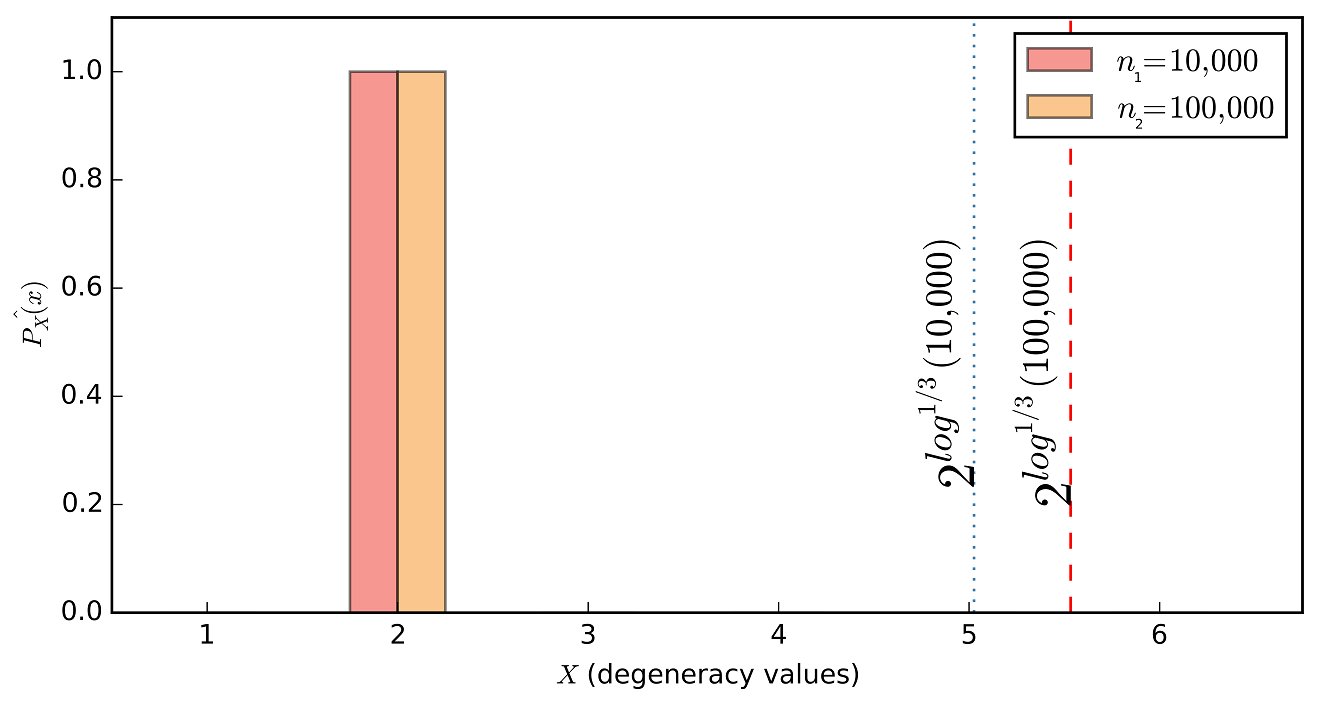}
	}
	
	\caption{
		The result of computing degeneracy on scale-free networks. Each of Figures \ref{fig:3.0}, \ref{fig:3.2}, \ref{fig:3.5}, and \ref{fig:4.0} is a relative frequency diagram of computing degeneracy on two groups of scale-free networks which the first and second groups are shown by red and orange colors respectively. The dash lines show the boundaries of $2^{log^{1/3}(10000)}$ and $2^{log^{1/3}(100000)}$ as well.
	}
	\label{fig:subfigures1}
\end{figure*}

For each $\beta \in \{3.0,3.3,3.6,3.9,4.2,4.5,4.8,5.1\}$, we have generated 1000 scale-free networks with 1000000 nodes by using the mentioned algorithm. Then, we remove the vertices with degree less than $\log{n}\log^{*}{n}$ from each network. To check if the degrees of vertices in the remaining network follows the power-law distribution, we use a statistical hypothesis testing, as follows, \cite{Clauset2009},
\begin{align*}
\begin{cases}
H_{0}: \text{data is generated from a power-law distribution}\\
H_{1}: \text{data is not generated from a power-law distribution}
\end{cases}
\end{align*}
As it is shown in Fig. \ref{figPvalues}, the obtained $p$-values are between $0.10$ and $1$, for all values of the parameters. such results are in favor of $H_{0}$, i.e. the networks follow the power-law distribution. It should be noted that, considering significant level $\leq 0.10$, the $p$-values must be greater than $0.10$ to acceptance the $H_{0}$ hypothesis.
For computing these $p$-values, we use the poweRlaw package (version $0.60.3$) that has written in $R$ software and has proposed with Gillespie \cite{Gillespie2015RPackage}.
\begin{figure}[!htb]
	\centering
	\includegraphics[scale=0.067]{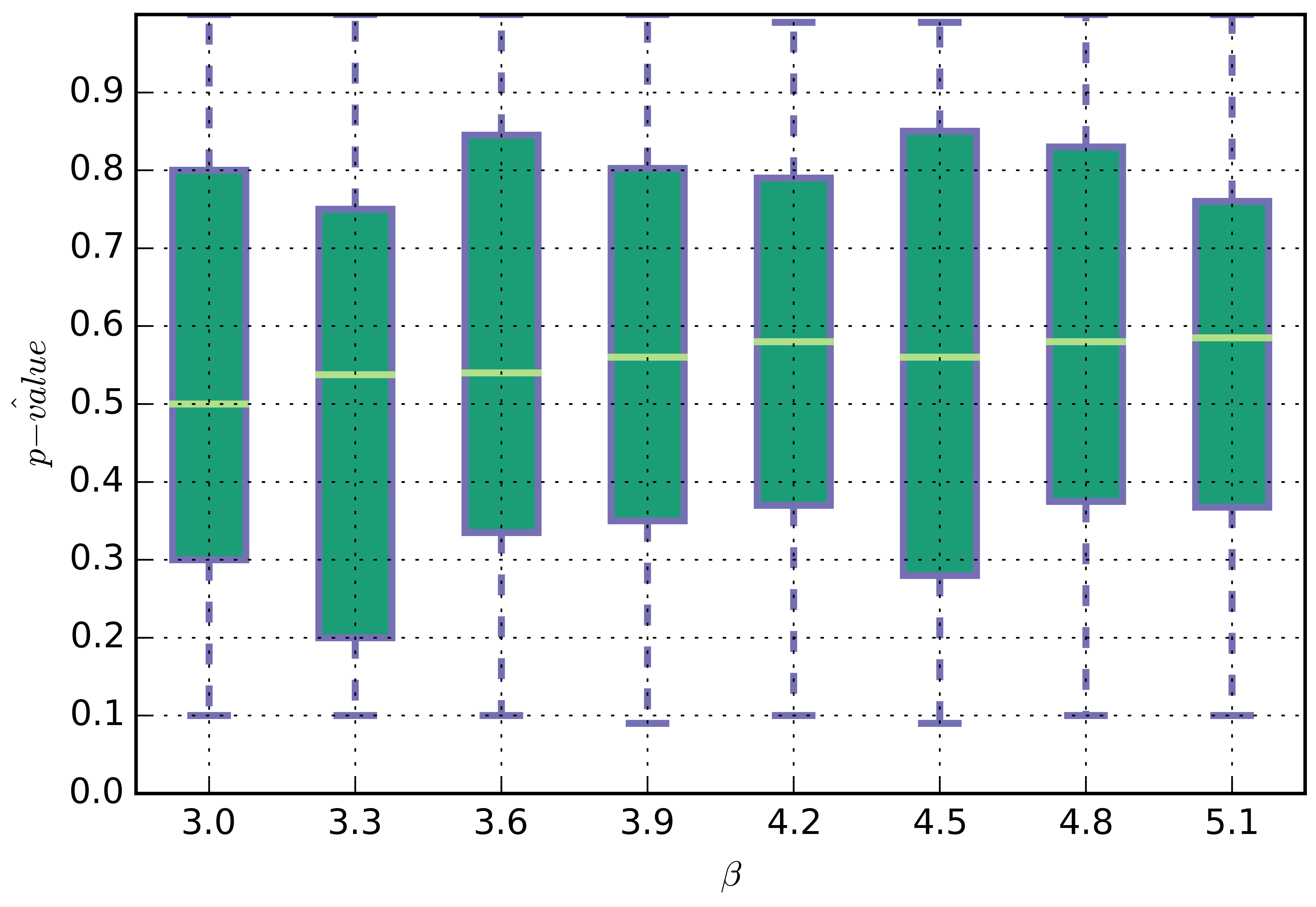}	
	\caption{The result of testing hypothesis about power-law. The distribution of degree for scale-free networks after removing vertices with degree less than $\log{n}\log^{*}{n}$ are computed. The $p$-values are shown by the box-plot.}\label{figPvalues}		
\end{figure}
\section{Conclusion and future works}
Two new algorithms with the time complexity of $O(\frac{\log{n}}{\log{\log{n}}})$ and $O(\log^{2/3}{n})$ rounds were presented for computing distributed MIS on scale-free networks. To this end, for modeling the scale-free networks, inhomogeneous random graphs with power-law weights were used. In addition, it was proved that the arboricity and degeneracy on these networks with power-law exponent $\beta \geq 3$ are less than $2^{log^{1/3} n}$ w.h.p. Hence, as the future work, it is a good idea to compute the arboricity and degeneracy of scale-free networks with power-law exponent $2\leq \beta < 3$. Moreover, we can work to propose an approach in order to solve $(\Delta+1)$-vertex coloring, finding maximal clique, and minimal dominating set on scale-free networks.

\bibliographystyle{tfnlm}
\bibliography{ref}
\end{document}